\newcommand{\full}[1]{#1}
\documentclass[11pt]{article}
\usepackage{geometry}
\pdfoutput=1 
\usepackage{verbatim}
\usepackage{float}
\usepackage{amsthm}
\usepackage{amsmath}
\usepackage{amssymb}
\usepackage{graphicx}
\usepackage[activate={true,nocompatibility},final,tracking=true,kerning=true,factor=1100,stretch=10,shrink=10]{microtype}
\usepackage{tikz}
\usetikzlibrary{calc}

\makeatletter

\floatstyle{ruled}
\newfloat{algorithm}{tbp}{loa}
\providecommand{\algorithmname}{Algorithm}
\floatname{algorithm}{\protect\algorithmname}

\theoremstyle{plain}
\newtheorem{thm}{\protect\theoremname}
  \theoremstyle{definition}
  
  \theoremstyle{definition}
  \newtheorem{defn}{\protect\definitionname}
  \theoremstyle{plain}
  \newtheorem{lem}{\protect\lemmaname}
  \newtheorem{fact}{Fact}
  
  \newtheorem{conj}{Conjecture}
  \newtheorem*{conj*}{Conjecture}
  \newtheorem*{conjrep*}{Conjecture~\ref{con:entropy} (repeated)}

\@ifundefined{date}{}{\date{}}
\usepackage{authblk}
\usepackage{algpseudocode}
\usepackage{algorithm}\usepackage{enumitem}

\newcommand{\expect}[1]{\mathbb{E}\left[{#1}\right]}

\usepackage{amsthm}
\usepackage{amsfonts}\usepackage{xspace}
\usepackage{upgreek}
\usepackage{braket}
\usepackage{enumitem}
\usepackage{amsfonts}

\renewcommand{\geq}{\geqslant}
\renewcommand{\leq}{\leqslant}
\renewcommand{\epsilon}{\varepsilon}
\renewcommand{\Delta}{\Updelta}

\newcommand{\calF}{\ensuremath{\mathcal{F}}}

\newcommand{\contrib}{\ensuremath{\textup{CT}}}

\newcommand{\expected}[1]{\ensuremath{\mathbb{E}[#1]}}
\newcommand{\expecteddisplay}[1]{\ensuremath{\mathbb{E}\left[#1\right]}}

\newcommand{\Xv}{\ensuremath{U_{\ell,t}}}
\newcommand{\Xvknown}{\ensuremath{\widetilde{U}_{\ell,t}}}
\newcommand{\Xvfix}{\ensuremath{\widetilde{u}_{\ell,t}}}
\newcommand{\Yv}{\ensuremath{A_{\ell,t}}}

\newcommand{\vitset}{\ensuremath{\mathcal{I}_{\ell,t}}}
\newcommand{\vitsetprime}{\ensuremath{\mathcal{I}_{\ell',t'}}}
\newcommand{\vitsetonlytprime}{\ensuremath{\mathcal{I}_{\ell,t'}}}
\newcommand{\git}{\ensuremath{G_{\ell,t}}}
\newcommand{\gitset}{\ensuremath{\mathcal{G}_{\ell,t}}}

\newcommand{\vit}{\ensuremath{I_{\ell,t}}}

\newcommand{\entropy}{\ensuremath{\kappa}}

\renewcommand{\S}{{\ensuremath{S}}} % Dynamic string (stream)
\newcommand{\U}{{\ensuremath{U}}} % Update string/sequence (text)
\newcommand{\arrive}{\ensuremath{\textsc{update}}}

\newcommand{\intlens}{\ensuremath{L}}
\newcommand{\arrivalset}{\ensuremath{T}}
\newcommand{\inta}{\ensuremath{t_0}}
\newcommand{\intb}{\ensuremath{t_1}}
\newcommand{\intc}{\ensuremath{t_2}}
\newcommand{\intd}{\ensuremath{t_3}}
\newcommand{\Fsub}[2]{\ensuremath{F_{#1}^{#2}}}
\newcommand{\Usub}[2]{\ensuremath{U_{#1}^{#2}}}

\newcommand{\insertdiagram}[1]{\includegraphics[scale=.8]{#1}}
\newcommand{\Patrascu}{P{\v a}tra{\c s}cu\xspace}

\newcommand{\st}{\,|\,}
\makeatother

  \providecommand{\definitionname}{Definition}
  \providecommand{\lemmaname}{Lemma}
  \providecommand{\problemname}{Problem}
\providecommand{\theoremname}{Theorem}

\begin{document}

%\pagenumbering{gobble}

\title{The complexity of computation in bit streams}

\author{Rapha\"{e}l Clifford \quad ~ Markus Jalsenius \quad ~ Benjamin Sach\\
 Department of Computer Science\\ University of Bristol\\ Bristol, UK}

\date{}

\maketitle

%\newpage
%\pagenumbering{arabic}
%\cleardoublepage
%\setcounter{page}{1}

%\pagenumbering{arabic}
%\setcounter{page}{1}%Leave this line commented out.

\begin{abstract}
We revisit the complexity of online computation in the cell probe model.  We consider a class of problems where we are first given a fixed pattern or vector $F$ of $n$ symbols and then one symbol arrives at a time in a stream. After each symbol has arrived we must output some function of $F$ and the $n$-length suffix of the arriving stream.  Cell probe bounds of $\Omega(\delta\lg{n}/w)$ have previously been shown for both convolution and Hamming distance in this setting, where $\delta$ is the size of a symbol in bits and $w\in\Omega(\lg{n})$ is the cell size in bits.  However, when $\delta$ is a constant, as it is in many natural situations, these previous results no longer give us non-trivial bounds. 

We introduce a new \emph{lop-sided information transfer} proof technique which enables us to prove meaningful lower bounds even for constant size input alphabets. We use our new framework to prove an amortised cell probe lower bound of $\Omega(\lg^2 n/(w\cdot \lg \lg n))$ time per arriving bit for an online version  of a well studied problem known as pattern matching with address errors.   This is the first non-trivial cell probe lower bound for any online problem on bit streams that still holds when the cell sizes are large. We also show the same  bound for online convolution conditioned on a new combinatorial conjecture related to Toeplitz matrices.
\end{abstract}

%%%%%%%%%%%%%%%%%%%%%%%%%%%%%%%%%%%%%%%%%%%%%%%%%%%%%%%%%%%%%%%%%%%%%%%%%%%%%%%%%%%%%%%%%
%%%%%%%%%%%%%%%%%%%%%%%%%%%%%%%%%%%%%%%%%%%%%%%%%%%%%%%%%%%%%%%%%%%%%%%%%%%%%%%%%%%%%%%%%
%%%%%%%%%%%%%%%%%%%%%%%%%%%%%%%%%%%%%%%%%%%%%%%%%%%%%%%%%%%%%%%%%%%%%%%%%%%%%%%%%%%%%%%%%
%%
%% Convolution
%%
%%%%%%%%%%%%%%%%%%%%%%%%%%%%%%%%%%%%%%%%%%%%%%%%%%%%%%%%%%%%%%%%%%%%%%%%%%%%%%%%%%%%%%%%%
%%%%%%%%%%%%%%%%%%%%%%%%%%%%%%%%%%%%%%%%%%%%%%%%%%%%%%%%%%%%%%%%%%%%%%%%%%%%%%%%%%%%%%%%%
%%%%%%%%%%%%%%%%%%%%%%%%%%%%%%%%%%%%%%%%%%%%%%%%%%%%%%%%%%%%%%%%%%%%%%%%%%%%%%%%%%%%%%%%%
%%%%%%%%%%%%%%%%%%%%%%%%%%%%%%%%%%%%%%%%%%%%%%%%%%%%%%%%%%%%%%%%%%%%%%%%%%%%%%%%%%%%%%%%%
%%%%%%%%%%%%%%%%%%%%%%%%%%%%%%%%%%%%%%%%%%%%%%%%%%%%%%%%%%%%%%%%%%%%%%%%%%%%%%%%%%%%%%%%%
\section{Introduction}
We revisit the complexity of online computation in the cell probe model. In recent years there has been considerable progress towards the challenging goal of establishing lower bounds for both static and dynamic data structure problems.  A third class of data structure problems, one which falls somewhere between these two classic settings, is online computation in a streaming setting. Here one symbol arrives at a time and a new output must be given after each symbol arrives and before the next symbol is processed. The key conceptual difference to a standard dynamic data structure problem is that although each arriving symbol can be regarded as a new update operation, there is only one type of query which is to output the latest value of some function of the stream.  

Online pattern matching is particularly suited to study in this setting and cell probe lower bounds have previously been shown  for different measures of distance including Hamming distance, inner product/convolution and edit distance~\cite{CJ:2011,CJS:2013, CJS-soda:2015}.  All these previous cell probe lower bounds have relied on only one proof technique, the so-called \emph{information transfer technique} of \Patrascu and Demaine~\cite{PD2006:Low-Bounds}.   In loose terms the basic idea is as follows. First one defines a random input distribution over updates.  Here we regard an arriving symbol as an update and after each update we perform one query which simply returns the latest distance between a predefined pattern and the updated suffix of the stream. Then one has to argue that knowledge of the answers to $\ell$ consecutive queries is sufficient to infer at least a constant fraction of the information encoded by $\ell$ consecutive updates that occurred in the past. If one can show this is true for all power of two lengths $\ell$ and ensure there is no double counting then the resulting lower bound follows by summing over all these power of two lengths.   

For the most natural cell size $w\in\Omega(\lg{n})$,   a cell probe lower bound of $\Omega(\delta \lg{n}/w)$ for both Hamming distance and convolution using this method was shown, where $\delta$ is the size of an input symbol, $w$ is the cell size in bits and $n$ is the length of the fixed pattern~\cite{CJ:2011,CJS:2013}.   When $\delta = w$, there is also a matching upper bound in the cell probe model and so no further progress is possible.  However, when the symbol size $\delta$ is in fact constant, the best lower bound that is derivable reduces trivially to be constant.  This is a particularly unfortunate situation as arguably the \emph{most} natural setting of parameters is when the input alphabet is of constant size but the cell size is not.

To make matters worse, this limitation is neither  specific to pattern matching problems nor even to online problems in general. As it is a fundamental feature of the information transfer technique the requirement to have large input alphabets also applies to a wide of dynamic data structure problems for which the information transfer technique has been up to this point the lower bound method of choice.   As a result we see the challenge of providing a new proof technique which can meaningfully handle constant sized alphabets as fundamental to the aim of advancing our knowledge of the true complexity of both online and dynamic problems.

We introduce a new lop-sided version of the information transfer technique that enables us to give meaningful lower bounds for precisely this setting, that is when $\delta \in O(1)$ and $w \in \Omega(\lg{n})$.  Our proof technique will rely on being able to show for specific problems that we need  only $\ell$ query answers to give at least a constant fraction of the information encoded in $\ell\lg{\ell}$ updates.  

We demonstrate our new framework by first by applying it to the well studied online convolution problem.  For this problem we give a conditional cell probe lower bound which depends on a new combinatorial conjecture involving Toeplitz matrices.  We then show that it is possible to derive an identical but this time unconditional lower bound by applying the lop-sided information transfer technique to a problem called online pattern matching with address errors~\cite{AABLLPSV:2009}.  This measure of distance arises in pattern matching problems where errors occur not in the content of the data but in the addresses for where the data is stored.

\subsection*{Previous cell probe lower bounds}

Our bounds hold in a particularly strong computational model, the \emph{cell-probe model}, introduced originally by Minsky and Papert~\cite{MP:1969} in a different context and then subsequently by Fredman~\cite{Fredman:1978} and Yao~\cite{Yao1981:Tables}. In this model, there is a separation between the computing unit and the memory, which is external and consists of a set of cells of $w$ bits each. The computing unit cannot remember any information between operations. Computation is free and the cost is measured only in the number of cell reads or writes (cell-probes). This general view makes the model very strong, subsuming for instance the popular word-RAM model. 

The first techniques known for establishing dynamic data structure lower bounds had historically been based on the chronogram technique of Fredman and Saks~\cite{FS1989:chronogram} which can at best give bounds of $\Omega(\lg{n}/\lg{\lg{n}})$. In 2004, \Patrascu and Demaine gave us  the first $\Omega(\lg{n})$ lower bounds for dynamic data structure problems~\cite{PD2006:Low-Bounds}. 
Their technique is based on information theoretic arguments which also form the basis for the work we present in this paper. \Patrascu and Demaine also presented ideas which allowed them to express more refined lower bounds such as trade-offs between updates and queries of dynamic data structures.  For a list of data structure problems and their lower bounds using these and related techniques, see for example~\cite{Pat2008:Thesis}.  More recently, a further breakthrough was made by Larsen who showed lower bounds of roughly $\Omega((\lg{n}/\lg{\lg{n}})^2)$ time per operation for dynamic weighted range counting problem and polynomial evaluation~\cite{Larsen:2012,Larsen:2012:focs}. These lower bounds remain the state of the art for any dynamic structure problem to this day.

% It will however also be instructive for our purposes to look at the case for edit distance and to ask why the same lower bound was not achieved in~\cite{CJS-soda:2015}.  In the online edit distance problem, the entropy of $\ell$ consecutive outputs is only linear in $\ell$ because consecutive outputs differ by at most a constant.  It is therefore not possible to infer a constant fraction of the information from $\ell$ updates if the symbols are drawn randomly from a non-constant alphabet. On the other hand, if the input alpabet is of constant size, that is $\delta$ is  a constant, then the resulting lower bound after summation becomes $\Omega(\lg{n}/w) = \Omega(1)$. As a result and in order to give a meaningful lower bound for edit distance the authors were compelled to abandon this most natural setting and to derive their results with single bit sized cells.  

\subsection{Our Results}

\paragraph{The lop-sided information transfer technique}

In the standard formulation of the information transfer technique of Demaine and \Patrascu~\cite{PD2004:Partial-sums}, adjacent time intervals are considered and the \emph{information} that is transferred from the operations in one interval of power of two length $\ell$ to the next interval of the same length is studied. It is sufficient to prove that a lower bound can be given for the information transfer which applies to all consecutive intervals of power of two length.  Conceptually a balanced tree on $n$ leaves over the time axis is constructed which is known as the information transfer tree.  An internal node $v$ is associated with the times $t_0$, $t_1$ and $t_2$ such that the two intervals $[t_0,t_1]$ and $[t_1+1,t_2]$ span the left subtree and the right subtree of $v$, respectively. By summing over the information transfer at each node $v$ the final lower bound is derived. 

In order to show a cell probe lower bound for bit streams, we will need to give lower bounds for the information transferred from intervals of time $[t_0,t_1]$ to later intervals $[t_2,t_3]$ which are shorter than the first intervals.  In particular we will want to argue about intervals of length $\ell \lg{\ell}$ and $\ell$ respectively. However, making the interval lengths lop-sided requires us to abandon the information transfer tree and indeed almost all of the previous proof technique.   We will instead place gaps in time between the end of one interval and the start of another and argue carefully both that not too much of the information can be lost in these gaps and that we can still sum the information transfer over a sufficient number of distinct interval lengths without too much double counting. 

Our hope is that this new technique will lead to a new class of cell probe lower bounds which could not be proved with existing methods.

\paragraph{Online convolution}

In the online convolution problem we are given a fixed vector $F \in \{0,1\}^n$ and a stream that arrives one bit at a time.  After each bit arrives we must output the inner product between $F$ and a vector formed from the most recent $n$-length suffix of the stream.  At the heart of our lower bound we need to show that we can recover $\Omega(\ell \lg \ell)$ bits of information of a contiguous subarray  of the stream of length $\ell \lg \ell$. However, because the inputs are binary and to avoid a trivial lower bound, we must do so using only $\ell$ outputs. As we need to achieve this for a large number distinct values of $\ell$ simultaneously, the fixed vector $F$ we design for the online convolution problem has a carefully designed recursive structure.  However even having overcome this hurdle, there is still the challenge of showing that for each separate value $\ell$, the information from successive outputs does not have too large an overlap.   A complete resolution to this question seems non-trivial  but in Section~\ref{sec:conv} we present a new combinatorial conjecture related to Toeplitz matrices, which if true, provides the desired lower bound. As a result we get the following lower bound for online convolution in a bit stream:

\begin{thm}[\textbf{Online convolution}]
    \label{thm:conv}
    Assuming Conjecture~\ref{con:entropy}, in the cell-probe model with $w\in\Omega(\lg n)$ bits per cell, for any randomised algorithm solving the online convolution problem on binary inputs there exist instances such that the expected amortised time per arriving value is
    \[
      \Omega{\left(\frac{\lg^2 n}{w\cdot \lg \lg n}\right)}.
    \]
\end{thm}

\paragraph{Online pattern matching with address errors ($L_2$-rearrangement distance)}

As our second example we give an explicit distance function for which we can now obtain the first unconditional online lower bound for symbol size $\delta = 1$.   We consider a problem setting known as pattern matching with address errors and in particular the $L_2$-rearrangement distance as defined in~\cite{AABLLPSV:2009}.  Consider two strings two strings $S_1$ and $S_2$  both of length $n$ where $S_2$ is a permutation of $S_1$. Now consider the set of permutations $\Pi$ so that for all $\pi \in \Pi$, $S_1[\pi(0), \dots, \pi(n-1)] = S_2$.  The $L_2$-rearrangement distance is defined to be $\min_{\pi\in\Pi}\sum_{j=0}^{n-1}|j-\pi(j)|^2$.  If $\Pi$ is empty, that is $S_2$ is in fact not a permutation of $S_1$, then the $L_2$-rearrangement distance is defined to be $\infty$. When considered as an offline problem, for a text of length $2n$ and a pattern of length $n$, Amir et al.\@ showed that the  $L_2$-rearrangement distance between the pattern and every $n$-length substring of the text could be computed in $O(n\lg {n})$ time~\cite{AABLLPSV:2009}. In the online $L_2$-rearrangement problem we are given a fixed pattern $F \in \{0,1\}^n$ and the stream arrives one bit at a time. After each bit arrives we must output the $L_2$-rearrangement distance between $F$ and the most recent $n$-length suffix of the stream. 

As before we need to recover $\Omega(\ell \lg \ell)$ bits of information of a contiguous sub-array  of the stream of length $\ell \lg \ell$.  Our technique allows us to recover $\Omega(\lg n)$ distinct bits of stream from each output. This is achieved by constructing $F$ and carefully choosing a highly structured random input distribution for the incoming stream in such a way that the contributions to the output from different regions of the stream have different magnitudes. We can then use the result to extract distinct information about the stream from different parts of each output.

Using this approach we get the following cell probe lower bound:

\begin{thm}[\textbf{Online $L_2$-rearrangement}]
    \label{thm:L2}
    In the cell-probe model with $w\in\Omega(\lg n)$ bits per cell, for any randomised algorithm solving the online $L_2$-rearrangement distance problem on binary inputs there exist instances such that the expected amortised time per arriving value is
    \[
      \Omega{\left(\frac{\lg^2 n}{w\cdot \lg \lg n}\right)}.
    \]
\end{thm}

\section{Lop-sided information transfer}\label{sec:preliminaries}

We will define the concept of information transfer, a particular set of cells probed by the algorithm, and explain how a bound on the size of the information transfer can be used when proving the overall lower bounds of Theorems~\ref{thm:conv} and~\ref{thm:L2}. All logarithms are in base two.

% For a positive integer $n$, $[n]$ denotes the set $\{0,\dots,n-1\}$.
% For an array $A$ of length $n$ and $i,j\in[n]$, we write $A[i]$
% to denote the value at position~$i$, and where $j\geq i$, $A[i,j]$
% denotes the $(j-i+1)$-length subarray of $A$ starting at position~$i$. All logarithms are in base two.
%We first introduce a unifying framework for the problems we consider.
%We differ discussion of integer multiplication within this framework until later and focus first on convolution and Hamming distance.

\subsection{Notation for the online problems}
The define some notation for our online problems.
There is a \emph{fixed~array}
$F\in\{0,1\}^n$ of length $n$ and an array $S\in\{0,1\}^n$ of length $n$, referred
to as the \emph{stream}.
%Both arrays have have length $n$.
%of integers,
%, referred to as the \emph{alphabet},
%and we let $\delta=\lfloor \lg q \rfloor$ denote the number of bits required to encode a value from $[q]$.
%The value $q$, or alternatively $\delta$, is a parameter of the problem.
%An element of
%the alphabet is often referred to as a \emph{symbol}.
We maintain $S$ subject to an update operation $\arrive(x)$ which takes a value $x\in \{0,1\}$, modifies~$S$ by appending $x$ to the right of the rightmost component $\S[n-1]$ and removing the leftmost component $S[0]$, and then outputs the inner product of $F$ and $S$, that is
$\sum_{i\in[n]}(F[i]\cdot S[i])$, or alternatively the $L_2$-rearrangement distance, depending on which problem we are currently considering.

We let $U\in \{0,1\}^n$ denote the \emph{update array} which describes a sequence of  $n$ $\arrive$ operations. That is, for each $t\in[n]$, the operation $\arrive(U[t])$ is performed.
We will usually refer to $t$ as the \emph{arrival} of the value $U[t]$.
%The choice of the letter $t$ is compatible with all previous work using the information transfer method and stands for time. Here we avoid using the term time as a measure of the number of values that have arrived in the stream.
Observe that just after arrival $t$, the values $U[t+1,n-1]$ are still not known to the algorithm. Finally, we let the $n$-length array $A$ denote the outputs such that for $t\in [n]$, $A[t]$ is the output of $\arrive(U[t])$.

%In the \emph{multiplication}~problem we let $F$ denote one of the two operands to be multiplied, hence $F$ is fixed and known in advance by the algorithm.
%Specifically we let $F[i]$ denote the $i$-th least significant digit. We let $U$ be the unknown operand so that $U[t]$ is its $t$-th least significant digit. Prior to the arrival of the first digit $U[0]$, the stream $S$ contains only zeros. The output $A[t]$ is the $t$-th digit in the product of $F$ and $S$, which is a function of $F$ and $U[0,t]$ as required.

\subsection{Hard distributions} \label{sec:hard-distributions}

Our lower bound holds for any randomised algorithm on its worst case
input. This will be achieved by applying \emph{Yao's
minimax principle}~\cite{Yao1977:Minimax}. We develop a
lower bound that
holds for any deterministic~algorithm on some random~input. The basic approach is as follows: we devise a fixed array $F$ and describe a probability
distribution for $U$, the $n$ new values arriving in the stream~$S$.
We then obtain a lower bound on the expected
running time over the arrivals in $U$ that holds for any deterministic algorithm.
Due to the minimax principle, the
same lower bound must then also hold for any randomised~algorithm on its own
worst~case input. The amortised bound per arriving value is obtained by dividing by $n$.

From this point onwards we consider an arbitrary deterministic algorithm
running with some fixed array $F$ on a random input of $n$ values. The algorithm may depend on~$F$. We refer to the choice of $F$ and distribution for $U$ as a \emph{hard distribution} since it used to show a lower bound.

%To obtain a randomised, \emph{worst-case} lower bound, it suffices to show the existence of such a hard distribution - in particular we give a single $F$ for each $n=|F|$. It is therefore natural to ask the question ``For which choices of $F$ does the lower bound hold?''. For the convolution problem and the integer multiplication problem, we answer this question by proving that our lower bounds hold even when $F$ is chosen uniformly at random. This is a slightly subtler claim and warrants clarification. Pick $F$ uniformly at random from $[q]^n$. Consider any Las Vegas randomised algorithm for convolution. This algorithm may depend on $F$.
%We then give a lower bound on the expected worst-case running time of this algorithm. The expectation is over the choice of $F$ while the worst-case input is over all sequences of $n$ stream updates. Therefore for most $F$, the lower bound holds.

\subsection{Two intervals and a gap} \label{sec:intervals}

In order to define the concept of information transfer from one interval of arriving values in the stream to another interval of arriving values, we first define the set $\intlens$ which contains the interval lengths that we will consider. The purpose of the next few definitions will be clear once we define information transfer in the next section. We let
\[
  \intlens =
    \Set{
      n^{1/4} \cdot (\lg n)^{2i}
      ~~|~~
      i \in \Big\{ 0,1,2,\dots, \frac {\lg n} {4 \lg \lg n} \Big\}
    }.
\]
\begin{comment}
For $\ell\in\intlens$ we define the set $\arrivalset_\ell\subseteq [n]$ of arrivals to be
\[
  \arrivalset_\ell =
  \Set{
    i\ell
    ~~\big|~~
    \textup{$i\in\{0,1,2,\dots\}$ and $i\ell\leq \frac n 2$}
  }.
\]
\end{comment}
To avoid cluttering the presentation with floors and ceilings, we assume throughout that the value of $n$ is such that any division or power nicely yields an integer. Whenever it is impossible to obtain an integer we assume that suitable floors or ceilings are used. In particular, $\intlens$ contains only integers.

For $\ell\in\intlens$ and $t\in[n/2]$ we define the following four values:
\begin{align*}
  \inta &= t,\\
  \intb &= \inta + \ell \lg \ell - 1,\\
  \intc &= \intb + \frac {4\ell} {\lg n} + 1,\\
  \intd &= \intc + \ell - 1.
\end{align*}
The values $\inta$, $\intb$, $\intc$ and~$\intd$ are indeed functions of $\ell$ and $t$ but for brevity we will often write just $\inta$ instead of $\inta(\ell,t)$, and so on, whenever the parameters $\ell$ and $t$ are obvious from context.
The four values define the three intervals
$[\inta,\intb]$, $[\intb+1,\intc-1]$ and $[\intc,\intd]$, referred to as the \emph{first interval}, the \emph{gap} and the \emph{second interval}, respectively.
Whenever there is a risk of ambiguity of which parameters the intervals are based on, we may specify what $\ell$ and $t$ they are associated with.
Before we explain the purpose of the intervals we will highlight some of their properties.

First observe that the intervals are disjoint and the gap has length $4\ell/\lg n$. The first interval has length $\ell\lg \ell$ and starts at $t$, where $t$ is always in the first half of the interval $[0,n-1]$. The second interval has length $\ell$, hence is a log-factor shorter than the first interval. All intervals are contained in $[0,n-1]$. To see this we need to verify that $\intd\leq n-1$ for any choice of $\ell\in\intlens$ and $t\in[n/2]$.
The largest value in $\intlens$ is
\[
  n^{1/4} \cdot (\lg n)^{\frac {2\lg n} {4 \lg \lg n}} = n^{3/4}.
\]
Thus, the largest possible value of $\intd$, obtained with $\ell=n^{3/4}$ and $t=n/2-1$, is
\begin{align*}
 \left( \frac n 2 - 1 \right)
 + \left( n^{3/4}\cdot \lg n^{3/4} - 1 \right)
 &+ \left( \frac {4n^{3/4}} {\lg n} + 1 \right)
 + \left( n^{3/4} - 1 \right)\\
 &\leq~
 \frac n 2 + 1 + n^{3/4} ( \lg n^{3/4} + 5 )
 ~<~
 n,
\end{align*}
whenever $n$ is sufficiently large.

Lastly, suppose that $\ell'\in\intlens$ is one size larger than $\ell\in\intlens$, that is $\ell'=\ell\cdot(\lg n)^2$.
For $\ell'$ the length of the gap is $4\ell'/\lg n$, which so big that it spans the length of both intervals plus the gap associated with $\ell$. To see this, observe that for sufficiently large~$n$,
\begin{align*}
  \intd(\ell,t)-\inta(\ell,t)+1
  &~=~
  \ell\lg\ell - 1 + \frac{4\ell}{\lg n} + 1 + \ell -1 + 1\\
  &~\leq~ 2\ell\lg n + \frac{4\ell}{\lg n}
  ~\leq~ 2\ell\lg n + \frac{(\lg n)^2\ell}{\lg n}
  ~=~ 3\ell\lg n
  ~\leq~ \frac{4\ell'}{\lg n}.
\end{align*}

\subsection{Information transfer over gaps} \label{sec:information-transfer}

Towards the definition of information transfer, we define, for $\ell\in\intlens$ and $t\in[n/2]$, 
the subarray $\Xv=U[\inta,\dots, \intb]$ to represent the $\ell\lg\ell$ values
arriving in the stream during the first interval.
We define the subarray $\Yv=A[\intc,\dots, \intd]$ to represent the $\ell$
outputs during the second interval.
Lastly we define $\Xvknown$ to be the concatenation of $U[0,(\inta-1)]$
and $U[(\intb+1),(n-1)]$. That is, $\Xvknown$ contains all values
of $U$ except for those in $\Xv$.
%When $\Xvknown$ is fixed to some constant $\Xvfix$ and $\Xv$ is
%random, we write $H(\Yv\mid\Xvknown=\Xvfix)$ to denote the conditional
%entropy of $\Yv$ under the fixed~$\Xvknown$.

For $\ell\in\intlens$ and $t\in[n/2]$ we first define the \emph{information transfer to the gap}, denoted $\gitset$, to be the set of memory cells $c$ such that $c$ is probed during the first interval $[\inta,\intb]$ of arriving values and also probed during the arrivals of the values $U[\intb+1,\intc-1]$ in the gap.
Similarly we define the information transfer to the second interval, or simply \emph{the information transfer}, denoted $\vitset$, to be the set of memory cells $c$ such that $c$ is probed during the first interval $[\inta,\intb]$ of arriving symbols and also probed during the arrivals of symbols in the second interval $[\intc,\intd]$ but \emph{not} in the gap. That is, any cell $c\in\gitset$ cannot also be contained in the information transfer $\vitset$.

The cells in
the information transfer $\vitset$ may contain
information about the values in $\Xv$ that the algorithm
uses in order to correctly produce the outputs $\Yv$.
However, since cells that are probed in the gap are not included in the information transfer, the information transfer might not contain all the information about the values in $\Xv$ that the algorithm uses while outputting $\Yv$.
In all previous work on lower bounds where the information transfer technique is used, the two intervals had no gap between them, hence the information transfer $\vitset$ contained \emph{all} the information about the updates in $\Xv$ necessary for computing $\Yv$. 
Further, in all previous work, the two intervals always had the same length.
We will see that the introduction of a gap and skewed interval lengths enable us to provide non-trivial lower bounds for small inputs with large cell sizes.
We will see that the gap is small enough that a sufficiently large portion of the information about $\Xv$ has to be fetched from cells in the information transfer $\vitset$.

Since cells in the information transfer are by definition probed at some point by the algorithm, we can use $\vitset$ to measure, or at least lower bound, the number of cell probes.
As a shorthand we let $\vit=|\vitset|$ denote the size of the information transfer~$\vitset$.
Similarly we let $\git=|\gitset|$ denote the size of the information transfer to the gap.
By adding up the sizes $\vit$ of the information transfers over all $\ell\in\intlens$ and certain values of $t\in[n/2]$, we get a lower bound on the total number of cells probed by the algorithm during the $n$ arriving values in $U$.
The choice of the values $t$ is crucial as we do not want to over-count the number of cell probes. In the next two lemmas we will deal with the potential danger of over-counting.

For a cell $c\in\vitset$, we write \emph{the probe of $c$ with respect to $\vitset$} to refer to the first probe of $c$ during the arrivals in the second interval. These are the probes of the cells in the information transfer that we count.

\begin{lem}
  \label{lem:over-counting-arrivals}
  For any $\ell\in\intlens$ and $t,t'\in[n/2]$ such that $|t-t'|\geq\ell$, if a cell $c$ is in both $\vitset$ and $\vitsetonlytprime$ then the probe of $c$ with respect to $\vitset$ and the probe of $c$ with respect $\vitsetprime$ are distinct.
\end{lem}
\begin{proof}
  Since $t$ and $t'$ are at least $\ell$ apart, the second intervals associated with $t$ and $t'$, respectively, must be disjoint.
  Hence the probe of $c$ with respect to $\vitset$ and the probe of $c$ with respect $\vitsetonlytprime$ must be distinct.
\end{proof}

From the previous lemma we know that there is no risk of over-counting cell probes of a cell over information transfers $\vitset$ under a fixed value of $\ell\in\intlens$, as long as no two values of $t$ are closer than~$\ell$.
In the next lemma we consider information transfers under different values of $\ell\in\intlens$.

\begin{lem}
  \label{lem:over-counting-length}
  For any $\ell,\ell'\in\intlens$ such that $\ell\neq\ell'$, and any $t,t'\in[n/2]$, if a cell $c$ is in both $\vitset$ and $\vitsetprime$ then the probe of $c$ with respect to $\vitset$ and the probe of $c$ with respect $\vitsetprime$ must be distinct.
\end{lem}
\begin{proof}
  Let $p$ be the probe of $c$ with respect to $\vitset$, and let $p'$ be the probe of $c$ with respect $\vitsetprime$. We will show that $p\neq p'$. Suppose without loss of generality that $\ell<\ell'$.
  From the properties of the intervals that were given in the previous section we know that the length of the gap associated with $\ell'$ is larger than the sum of lengths of the first interval, the gap and the second interval associated with $\ell$.
  
   Suppose for contradiction that $p=p'$. By definition of $\vitset$, the cell $c$ is probed also in the first interval associated with $\ell$. Let $p_\textup{first}$ denote any such cell probe.
   Because the gap associated with $\ell'$ is so large, $p_\textup{first}$ must take place either in the second interval or the gap associated with $\ell'$. If $p_\textup{first}$ is in the gap, then $c$ cannot be in $\vitsetprime$. If $p_\textup{first}$ is in the second interval then $p'$ cannot equal $p$.
\end{proof}

In order to lower bound the total number of cell probes performed by the algorithm over the $n$ arrivals in $U$ we will define, for each $\ell\in\intlens$, a set $\arrivalset_\ell\subseteq [n/2]$ of arrivals, such that for any distinct $t,t'\in\arrivalset_\ell$, $|t-t'|\geq\ell$.
It then follows from Lemmas~\ref{lem:over-counting-arrivals} and~\ref{lem:over-counting-length} that
\[
      \sum_{\ell\in\intlens}
      \sum_{t\in\arrivalset_\ell}
      \vit
\]
is a lower bound on the number of cell probes. Our goal is to lower bound the expected value of this double-sum.
The exact definition of $\arrivalset_\ell$ will be given in Section~\ref{sec:lower-it} once we have introduced relevant notation.

\section{Proving the lower bound}\label{sec:proofs}

In this section we give the overall proof for the lower bounds of Theorems~\ref{thm:conv} and~\ref{thm:L2}.
Let $\ell\in\intlens$ and let $t\in[n/2]$.
%Let $\ell\in\intlens$ and let $\arrivalset_\ell$ equal any quick set $\arrivalset_{\ell,f}$ for some $f\in[\ell]$, which we know exists due to Lemma~\ref{lem:quick}.
%Let $t\in\arrivalset_\ell$ and suppose that $\Xvknown$ is fixed but the
Suppose that $\Xvknown$ is fixed but the
values in $\Xv$ are drawn at random in accordance with the distribution
for $U$, conditioned on the fixed value of $\Xvknown$.
This induces a distribution for the outputs $\Yv$.
We want to show that if the entropy of $\Yv$ is
large, conditioned on the fixed $\Xvknown$, then
the information transfer $\vitset$ is large, since only the variation in the inputs $\Xv$ can alter the outputs $\Yv$.
We will soon make this claim more precise.

\subsection{Upper bound on the entropy}

We write $H(\Yv\mid\Xvknown=\Xvfix)$ to denote the entropy of $\Yv$ conditioned on fixed $\Xvknown$.
Towards showing that high conditional entropy  $H(\Yv\mid\Xvknown=\Xvfix)$ implies large information transfer we use the information transfer $\vitset$ and the information transfer to the gap, $\gitset$, to describe an encoding of the outputs $\Yv$. The following lemma gives a direct relationship between $\vit+\git$ and the entropy.
A marginally simpler version of the lemma, stated with different notation, was first given in~\cite{PD2006:Low-Bounds} under the absence of gaps.

\begin{lem}
%[\Patrascu and Demaine~\cite{PD2006:Low-Bounds}]
\label{lem:H-upper-old}
Under the assumption that the address of any cell can be specified in $w$ bits,
for any $\ell\in\intlens$ and $t\in[n/2]$, the entropy 
$$
H(\Yv\mid\Xvknown=\Xvfix)~\leq~2w + 2w\cdot \expected{\vit +\git \mid \Xvknown=\Xvfix}.
$$
\end{lem}
\begin{proof}
  The expected length of any encoding of $\Yv$ under fixed $\Xvknown$ is an upper bound on the conditional entropy of $\Yv$.
  We use the information transfer $\vitset$ and the information transfer to the gap, $\gitset$, to define an encoding of $\Yv$ in the following way. For every cell $c\in\vitset\cup\gitset$ we store the address of $c$, which takes at most $w$ bits under the assumption that a cell can hold the address of any cell in memory.
  We also store the contents of $c$ that it holds at the very end of the first interval, just before the beginning of the gap. The contents of $c$ is specified with $w$ bits.
  In total this requires $2w\cdot (\vit+\git)$ bits.

  We will use the algorithm, which is fixed, and the fixed values $\Xvfix$ of $\Xvknown$ as part of the decoder to obtain $\Yv$ from the encoding. Since the encoding is of variable length we also store the size $\vit$ of the information transfer and the size $\git$ of the information transfer to the gap. This requires at most $2w$ additional bits.

  In order to prove that the described encoding of $\Yv$ is valid we now describe how to decode it.
  First we simulate the algorithm on the fixed input $\Xvknown$ from the first arrival $U[0]$ until just before the first interval when the first value in $\Xv$ arrives.
  We then skip over all inputs in $\Xv$
  and resume simulating the algorithm from the beginning of the gap, that is when the value $U[\intb+1]$ arrives.
  We simulate the algorithm over the arrivals in the gap and the second interval until all values in $\Yv$ have been outputted.
  For every cell being read, we check if it is contained in either the information transfer $\vitset$ or the information transfer to the gap $\gitset$ by looking up its address in the encoding.
  If the address is found then the contents of the cell is fetched from the encoding. If not, its contents is available from simulating the algorithm on the fixed inputs $\Xvknown$.
\end{proof}

\subsection{Lower bounds on entropy}

Lemma~\ref{lem:H-upper-old} above provides a direct way to obtain a lower bound on the expected value of $\vit+\git$ if given a lower bound on the conditional entropy $H(\Yv\mid\Xvknown=\Xvfix)$.
In the next two lemmas we provide such entropy lower bounds for $L_2$-rearrangement distance and convolution.

\begin{lem}
  \label{lem:conv-H-lower}
  Assuming Conjecture~\ref{con:entropy}, for the convolution problem there exists a real constant $\entropy>0$ and, for any $n$, a fixed array $F\in \{0,1\}^n$ such that for all $\ell\in\intlens$ and $t\in[n/2]$, when $U$ is chosen uniformly at random from $\{0,1\}^n$ then
  \[
    H(\Yv\mid\Xvknown=\Xvfix)~\geq~ \entropy\cdot \ell\cdot \lg n,
  \]
  for any fixed $\Xvfix$.
\end{lem}

\begin{lem}
  \label{lem:L2-H-lower}
  For the $L_2$-rearrangement distance problem there exists a real constant $\entropy>0$ and, for any $n$, a fixed array $F\in \{0,1\}^n$ such that for all $\ell\in\intlens$ and all $t\in[n/2]$ such that $t \bmod 4 =0$, when $U$ is chosen uniformly at random from $\{0101,1010\}^{\frac{n}{4}}$ then
  \[
    H(\Yv\mid\Xvknown=\Xvfix)~\geq~ \entropy\cdot \ell\cdot \lg n,
  \]
  for any fixed $\Xvfix$.
\end{lem}

The proof of Lemma~\ref{lem:conv-H-lower} is deferred to Section~\ref{sec:conv} and hinges on a conjecture relating to Toeplitz matrices. The proof of Lemma~\ref{lem:L2-H-lower} is deferred to Section~\ref{sec:L2}. 
%, where we also explain how it forms the basis of the hard distribution of Lemma~\ref{lem:conv-H-lower}.

Before we proceed with the lower bound on the information transfer we make a short remark on the bounds that these lemmas give.
Observe that the maximum conditional
entropy of $\Yv$ is bounded by the entropy of $\Xv$, which is
$O(\ell\lg \ell)$ since the length of the first interval is $\ell \lg \ell$.
%are independent and uniformly drawn from $\{0,1\}$.
Recall also that the values in $\intlens$ range from $n^{1/4}$ to $n^{3/4}$.
Thus, for a constant~$\entropy$, both entropy lower bounds are tight up to a multiplicative constant factor.

\subsection{A lower bound on the information transfer and quick gaps}
\label{sec:lower-it}

In this section we prove our main lower bound results. In order to fix ideas we do this for the $L_2$-rearrangement problem, given in Theorem~\ref{thm:L2}. As a result we assume that $\entropy$ is the constant and $F$ is the fixed array of Lemma~\ref{lem:conv-H-lower}, and that $U$ is chosen uniformly at random from $\{0,1\}^{n}$. Theorem~\ref{thm:conv} which gives our main lower bound result for online convolution follows via exactly the same argument but with Lemma~\ref{lem:conv-H-lower} replaced by Lemma~\ref{lem:L2-H-lower}. 

By combining the upper and lower bounds on the conditional entropy from Lemmas~\ref{lem:H-upper-old} and~\ref{lem:L2-H-lower} we have that there is a hard distribution and a real constant $\entropy>0$ for the convolution problem such that
\[
\expected{\vit +\git \mid \Xvknown=\Xvfix}
~\geq~
\frac {\entropy\cdot \ell\cdot \lg n} {2w} - 1
\]
for any $\Xvfix$.
We may remove the conditioning by taking expectation over $\Xvknown$ under random $U$. Thus,
\begin{equation*}
  \expected{\vit +\git}
  ~\geq~
  \frac {\entropy\cdot \ell\cdot \lg n} {2w} - 1,
\end{equation*}
or equivalently
\begin{equation}
  \label{eq:vitgit}
  \expected{\vit}
  ~\geq~
  \frac {\entropy\cdot \ell\cdot \lg n} {2w} - 1 - \expected{\git}.
\end{equation}
Recall that our goal is to lower bound
\[
      \expecteddisplay{
	\sum_{\ell\in\intlens}
	\sum_{t\in\arrivalset_\ell}
	\vit
      }
      ~=~
      \sum_{\ell\in\intlens}
      \sum_{t\in\arrivalset_\ell}
      \expect{\vit},
\]
where is $\arrivalset_\ell$ contains suitable values of $t$.
Using Inequality~(\ref{eq:vitgit}) above would immediately provide such a lower bound, however, there is an imminent risk that the negative terms of $\expected{\git}$ could devalue such a bound into something trivially small.
%The information transfer $\vitset$ holds information about the inputs during the first interval that the algorithm uses in order to produce the outputs during the second interval. Since there is a gap between the intevals, some of the information about the inputs may instead be contained in the information transfer to the gap, $\gitset$. Since we only count the number of cells in the information transfers $\vitset$ it is important to ensure that $\vitset$ remains sufficiently large.
%One can imagine an algorithm that spends a long time in the gap, reading cells that contain information about the updates from the first interval and re-writing them to new cells, thereby cirumventing a large information transfer $\vitset$ at the expence of a large information transfer to the gap.
Now, for this to happen, the algorithm must perform sufficiently many cell probes in the gap.
Since the length of the gap is considerably shorter than the second interval, a cap on the worst-case number of cell probes per arriving value would certainly ensure that $\expected{\git}$ stays small, but as we want an amortised lower bound we need something more refined.
The answer lies in how we define $\arrivalset_\ell$.
We discuss this next.

For $\ell\in\intlens$ and $f\in [\ell]$ we first define the set $\arrivalset_{\ell,f}\subseteq [n]$ of arrivals to be
\[
  \arrivalset_{\ell,f} =
  \Set{
    f+i\ell
    ~~|~~
    \textup{$i\in\{0,1,2,\dots\}$ and $f+i\ell\leq \frac n 2$}
  }.
\]
The values in $\arrivalset_{\ell,f}$ are evenly spread out with distance $\ell$, starting at $f$. We may think of $f$ as the offset of the sequence of values in $\arrivalset_{\ell,f}$. The largest value in the set is no more than $n/2$.
We will define the set $\arrivalset_\ell$ to equal a subset of one of the sets $\arrivalset_{\ell,f}$ for some $f$.
More precisely, we will show that there must exist an offset $f$ such that at least half of the values $t\in\arrivalset_{\ell,f}$ have the property that the time spent in the gap associated with $\ell$ and $t$ is small enough to ensure that the information transfer to the gap is small.
We begin with some definitions.

\begin{defn}[\textbf{Quick gaps and sets}]
  For any $\ell\in\intlens$ and $t\in[n/2]$ we say that the gap associated with $\ell$ and $t$ is \emph{quick} if the expected number of cell probes during the arrivals in the gap is no more than
  \[
    \frac {\entropy\cdot \ell\cdot \lg n} {4w},
  \]
  where $\entropy$ is the constant from Lemma~\ref{lem:conv-H-lower}.
  Further, for any $f\in[\ell]$ we say that the set $\arrivalset_{\ell,f}$ is \emph{quick} if, for at least half of all $t\in\arrivalset_{\ell,f}$, the gap associated with $\ell$ and $t$ is quick.
\end{defn}

The next lemma says that for sufficiently fast algorithms there is always an offset $f$ such that $\arrivalset_{\ell,f}$ is quick.

\begin{lem}
  \label{lem:quick}
  Suppose that the expected total number of cell probes over the $n$ arrivals in $U$ is less than
  \[
    \frac {\entropy\cdot n\cdot \lg^2 n} {32w}.
  \]
  %where $\entropy$ is the constant from Lemma~\ref{lem:conv-H-lower}.
  Then, for any $\ell\in\intlens$, there is an $f\in[\ell]$ such that $\arrivalset_{\ell,f}$ is quick.
\end{lem}

\begin{proof}
  In accordance with the lemma, suppose that the expected total number of cell probes over the $n$ arrivals in $U$ is less than
  $\entropy n (\lg^2 n) / (32w)$.
  For contradiction, suppose that there is no $f\in [\ell]$ such that $\arrivalset_{\ell,f}$ is quick.
  We will show that the expected number of cell probes over the $n$ arrivals must then be at least $\entropy n (\lg^2 n) / (32w)$.
  
  For any $f\in[\ell]$, let $R_f\subseteq [n]$ be the union of all arrivals that belong to a gap associated with $\ell$ and any $t\in\arrivalset_{\ell,f}$.
  Let $P_f$ be the number of cell probes performed by the algorithm over the arrivals in $R_f$.  
  Thus,
  for any set $\arrivalset_{\ell,f}$ that is \emph{not quick} we have by linearity of expectation
  \[
    \expect{P_f}
    ~\geq~
    \frac {|\arrivalset_{\ell,f}|} 2 \cdot
	\frac {\entropy\cdot \ell\cdot \lg n} {4w}
    ~=~
    \frac {n/2} {2\ell} \cdot
	\frac {\entropy\cdot \ell\cdot \lg n} {4w}
    ~=~
    \frac {\entropy\cdot n\cdot \lg n} {8w}.
  \]
  
  Let the set of offsets
  \[
    \calF = 
      \Set{
	i\cdot \frac {4\ell}{\lg n}
	~~|~~
	i \in \left[\frac {\lg n} 4 \right]
      }
      ~\subseteq~
      [\ell].
  \]
  The values in $\calF$ are spread out with distance $4\ell/\lg n$, which equals the gap length.
  Thus, for any two distinct $f,f'\in\calF$, the sets $R_f$ and $R_{f'}$ are disjoint.
  We therefore have that the total running time over all $n$ arrivals in $U$ must be lower bounded by $\sum_{f\in\calF} P_f$.
  Under the assumption that no $\arrivalset_{\ell,f}$ is quick, we have that the expected total running time is at least
  \[
    \expecteddisplay{ \sum_{f\in\calF} P_f }
    ~=~
    \sum_{f\in\calF} \expect{P_f}
    ~\geq~
    |\calF|\cdot \frac {\entropy\cdot n\cdot \lg n} {8w}
    ~=~
    \frac {\lg n} {4} \cdot \frac {\entropy\cdot n\cdot \lg n} {8w}
    ~=~
    \frac {\entropy\cdot n\cdot \lg^2 n} {32w},
  \]
  which is the contradiction we wanted.
  Thus, under the assumption that the running time over the $n$ arrivals in $U$ is less than
  $\entropy n (\lg^2 n) / (32w)$
  there must be an $f\in [\ell]$ such that $\arrivalset_{\ell,f}$ is quick.
\end{proof}

We now proceed under the assumption that the expected running time over the $n$ arrivals in $U$ is less than $\entropy n (\lg^2 n) / (32w)$. If this is not the case then we have already established the lower bound of Theorem~\ref{thm:conv}.

Let $f$ be a value in $[\ell]$ such that $\arrivalset_{\ell,f}$ is a quick set. Such an $f$ exists due to Lemma~\ref{lem:quick}.
We now let $\arrivalset_\ell\subseteq\arrivalset_{\ell,f}$ be the set of all $t\in\arrivalset_{\ell,f}$ for which the gap associated with $\ell$ and $t$ is quick. Hence
$|\arrivalset_\ell|\geq |\arrivalset_{\ell,f}|/2=n/(4\ell)$.
Since $\git$ cannot be larger than the number of cell probes in the gap, we have by the definition of a quick gap that for any $t\in\arrivalset_\ell$,
\[
  \expect{\git}
  \leq
  \frac {\entropy\cdot \ell\cdot \lg n} {4w}. 
\]
Using the above inequality together with Inequality~(\ref{eq:vitgit}) we can finally provide a non-trivial lower bound on the sum of the information transfers:
\begin{align*}
    \sum_{\ell\in\intlens}
    \sum_{t\in\arrivalset_\ell}
    \expect{\vit}
  &~\geq~
    \sum_{\ell\in\intlens}
    \sum_{t\in\arrivalset_\ell}
    \left(
      \frac {\entropy\cdot \ell\cdot \lg n} {2w} - 1 - \expected{\git}
    \right)\\
  &~\geq~
    \sum_{\ell\in\intlens}
    \sum_{t\in\arrivalset_\ell}
    \left(
      \frac {\entropy\cdot \ell\cdot \lg n} {2w} - 1
	- \frac {\entropy\cdot \ell\cdot \lg n} {4w}
    \right)\\
  &~\geq~
    \frac {\entropy\cdot \lg n} {5w}
    \sum_{\ell\in\intlens}
    \sum_{t\in\arrivalset_\ell}
      \ell
    ~\geq~
    \frac {\entropy\cdot \lg n} {5w}
    \sum_{\ell\in\intlens}
      \left(|\arrivalset_\ell|\cdot \ell\right)
    ~\geq~
    \frac {\entropy\cdot \lg n} {5w}
    \sum_{\ell\in\intlens}
      \left(\frac {n} {4\ell}\cdot \ell\right)\\      
  &~=~
    \frac {\entropy\cdot n\cdot \lg n} {20w}
    \cdot |\intlens|
    ~\geq~
    \frac {\entropy\cdot n\cdot \lg n} {20w}
    \cdot \frac {\lg n} {4 \lg \lg n}
    ~\in~
    \Theta{\left(\frac {n\cdot \lg^2 n} {w\cdot \lg \lg n} \right)}.
\end{align*}
By Lemmas~\ref{lem:over-counting-arrivals} and~\ref{lem:over-counting-length} this lower bound is also a bound on the expected total number of cell probes performed by the algorithm over the $n$ arrivals in $U$.
By Yao's minimax principle, as discussed in Section~\ref{sec:hard-distributions}, this implies that any randomised algorithm on its worst case input has the same lower bound on its expected running time.
The amortised time per arriving value is obtained by dividing the running time by $n$.
This concludes the proofs of Theorem~\ref{thm:L2} and~\ref{thm:conv} .

\section{The hard distribution for convolution} \label{sec:conv}

In this section we prove Lemma~\ref{lem:conv-H-lower}, which says that assuming Conjecture~\ref{con:entropy} holds, there exists a real constant $\entropy>0$ and a fixed array $F\in \{0,1\}^n$ such that for all $\ell\in\intlens$ and $t\in[n/2]$, when $U$ is chosen uniformly at random from $\{0,1\}^n$ then the conditional entropy of the outputs $\Yv$ is
\[
  H(\Yv\mid\Xvknown=\Xvfix)~\geq~ \entropy\cdot \ell\cdot \lg n,
\]
for any fixed $\Xvfix$.
We begin by discussing the fixed array $F$.

\subsection{The array $F$} \label{sec:F}

For each $\ell\in\intlens$ there is a subarray of $F$ of length $\ell\lg\ell+\ell-1$.
Each such subarray, which we denote $F_\ell$, is at distance $4\ell/\lg n+1$ from the right-hand end of $F$, which is one more than the length of the gap associated with $\ell$.
Figure~\ref{fig:F} illustrates two subarrays $F_\ell$ and $F_{\ell'}$, where $\ell,\ell'\in\intlens$ and $\ell<\ell'$.
\begin{figure*}[t]
    \centering
    \insertdiagram{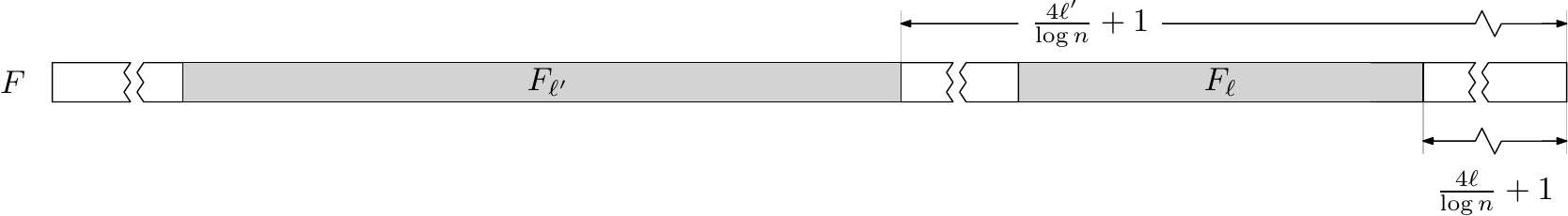}
    \caption{\label{fig:F}Two subarrays $F_\ell$ and $F_{\ell'}$ of $F$. There is never an overlap between any two such subarrays.}
\end{figure*}
By the properties discussed in Section~\ref{sec:intervals} we know that the length of the gap associated with $\ell'$ is larger than the length of $F_\ell$ plus the length of the gap associated with $\ell$. Hence there is no overlap between the subarrays $F_\ell$ and~$F_{\ell'}$.

Given any of the subarrays $F_\ell$ and an array $U_\ell$ of length $\ell\lg \ell$, we write $F_\ell\otimes U_\ell$ to denote the $\ell$-length array that consists of all inner products between $U_\ell$ and every substring of $F_\ell$. More precisely, for $i\in[\ell]$ the $i$-th component of $F_\ell\otimes U_\ell$ is
\[
  (F_\ell \otimes U_\ell)[i]
  =
  \sum_{j\in[\ell\lg\ell]}
    \big(F_\ell[i+j] \cdot U_\ell[j]\big).
\]
As Figure~\ref{fig:slide} shows we may think of $F_\ell\otimes U_\ell$ as the inner products of $U_\ell$ and its aligned subarray of $F_\ell$ as $U_\ell$ slides along $F_\ell$.
\begin{figure*}[t]
    \centering
    \insertdiagram{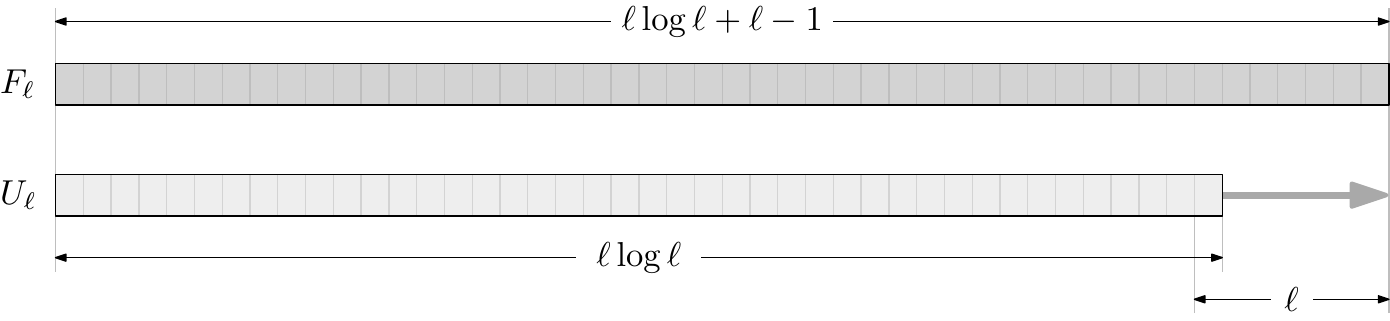}
    \caption{\label{fig:slide}The $\ell$-length array $F_\ell\otimes U_\ell$ contains the inner products of $U_\ell$ and corresponding subarrays of $F_\ell$ as $U_\ell$ slides along $F_\ell$.}
\end{figure*}
We can show that there exist subarrays $F_\ell$ such that the entropy of $F_\ell\otimes U_\ell$ is high when $U_\ell$ is drawn uniformly at random from $\{0,1\}^{\ell\lg\ell}$.
Observe that the entropy of $F_\ell\otimes U_\ell$ is upper bounded by the entropy of $\U_\ell$, which is exactly $\ell\lg\ell$.
The proof of the next lemma is based on a conjecture that we describe shortly.

\begin{lem}
  \label{lem:subarray}
  There exists a real constant $\epsilon>0$ such that for all $n$ and $\ell\in\intlens$ there is a subarray $F_\ell$ for which the entropy of $F_\ell\otimes U_\ell$ is at least $\epsilon\cdot\ell\lg\ell$ when $U_\ell$ is drawn uniformly at random from $\{0,1\}^{\ell\lg\ell}$.
\end{lem}

In order to finish the description of the array $F$ we choose each subarray $F_\ell$ such that the entropy of $F_\ell\otimes U_\ell$ is at least $\epsilon\cdot\ell\lg\ell$ when $U_\ell$ is drawn uniformly at random from $\{0,1\}^{\ell\lg\ell}$, where $\epsilon$ is the constant of Lemma~\ref{lem:subarray}. Any element of $F$ that is not part of any of the subarrays $F_\ell$ is chosen arbitrarily. This concludes the description of the array $F$.

\subsection{Proving the entropy lower bound}

We are now ready to prove Lemma~\ref{lem:conv-H-lower}, that is lower bound the conditional entropy of $\Yv$.

\begin{proof}[Proof of Lemma~\ref{lem:conv-H-lower}]
  Let $F$ be the array described above and let $U$ be drawn uniformly at random from $\{0,1\}^n$. Let $\ell\in\intlens$ and $t\in[n/2]$.
  Thus, conditioned on any fixed $\Xvknown$, the distribution of $\Xv$ is uniform on $\{0,1\}^{\ell\lg\ell}$.
  
  Recall that $\Xv$ arrives in the stream between arrival $\inta$ and $\intb$, after which $4\ell/\lg n$ values arrive in the gap. Thus, at the beginning of the second interval, at arrival $\intc$, $\Xv$ is aligned with the $(\ell\lg\ell)$-length suffix of the subarray $F_\ell$ of $F$. Over the $\ell$ arrivals in the second interval, $\Xv$ slide along $F_\ell$ similarly to Figure~\ref{fig:slide} (only in reversed direction of what the diagram shows). Since all values in $\Xvknown$ are fixed, the outputs $\Yv$ uniquely specify $F_\ell\otimes\Xv$. Thus, by Lemma~\ref{lem:subarray}, the conditional entropy
  \[
    H(\Yv\mid\Xvknown=\Xvfix)
    ~\geq~
    \epsilon\cdot \ell\cdot \lg \ell,
    ~\geq~
    \frac{\epsilon}{4}\cdot \ell\cdot \lg n,
  \]
  since $\ell\geq n^{1/4}$. By setting the constant $\entropy$ to $\epsilon/4$ we have proved Lemma~\ref{lem:conv-H-lower}.
\end{proof}

The last piece remaining is the proof of Lemma~\ref{lem:subarray}. This proof is based on a conjecture which we state as follows. Recall that a \emph{Toeplitz} matrix (or upside-down \emph{Hankel} matrix) is constant on each descending diagonal from left to right. 

\newcommand{\conjdef}{There exist two positive real constants $\alpha\leq 1$ and $\gamma\leq 1$ such that for any $h$ there is a $(0,1)$-Toeplitz matrix $M$ of height $h$ and width $\alpha\cdot h\lg h$ with the property that the entropy of the product $Mv$ is at least $\gamma\cdot h\lg h$, where $v$ is a column vector of length $\alpha\cdot h\lg h$ whose elements are chosen independently and uniformly at random from $\{0,1\}$.}

\begin{conj}
  \label{con:entropy}
  \conjdef
\end{conj}

This conjecture might at first seem surprising as the matrix is non-square. However, an essentially equivalent statement  was shown to be true for general $(0,1)$-matrices in 1974 by Erd\H{o}s and Spencer~\cite{ES:1974}.  Moreover, even stronger statements for general $(0,1)$-matrices form the basis of the well studied ``Coin Weighing Problem with a Spring Scale'' (see~\cite{Bshouty:2009} and references therein).  We will now show how to prove Lemma~\ref{lem:subarray} using the conjecture.

\begin{proof}[Proof of Lemma~\ref{lem:subarray} (assuming Conjecture~\ref{con:entropy})]
  Let $\alpha$ and $\gamma$ be the two constants in Conjecture~\ref{con:entropy}.
  Let $h=\ell$ and let $M$ be a $(0,1)$-Toeplitz matrix of height $\ell$ and width $\alpha\cdot \ell\lg \ell$ with the property of Conjecture~\ref{con:entropy}.
  
  We now define a new matrix $M_\ell$ of height $\ell$ and width $\ell\lg \ell$ such that the submatrix of $M_\ell$ that spans the first $\alpha\cdot \ell\lg \ell$ columns equals $M$.
  Remaining elements of $M_\ell$ are filled in arbitrarily from the set $\{0,1\}$ so that   $M_\ell$ becomes Toeplitz. 
  
  Let $v_\ell$ be a random column vector of length $\ell\lg \ell$ such that the first $\alpha\cdot \ell\lg \ell$ elements are chosen independently and uniformly from $\{0,1\}$. The remaining elements are fixed arbitrarily.
  By Conjecture~\ref{con:entropy} we have that the entropy of $M_\ell v_\ell$ is at least $\gamma\cdot \ell\lg \ell$.
  Thus, if we instead pick all elements of $v_\ell$ independently and uniformly at random from $\{0,1\}$ then the conditional entropy of $M_\ell v_\ell$, conditioned on all but the first $\alpha\cdot \ell\lg \ell$ elements, is at least $\gamma\cdot \ell\lg \ell$. Hence the entropy of $M_\ell v_\ell$ is also at least $\gamma\cdot \ell\lg \ell$.
  
  Now, since $M_\ell$ is Toeplitz we have that the first column and the first row of $M_\ell$ define the entire matrix. These $\ell + \ell\lg \ell - 1$ elements can be represented with an array $F_\ell$ of length $\ell + \ell\lg \ell - 1$ such that $F_\ell\otimes v_\ell=M_\ell v_\ell$, where we abuse notation slightly. In other words, the elements of the product $M_\ell v_\ell$ correspond exactly to the inner products obtained by sliding $U_\ell$, where $U_\ell$ is the array version of the column vector $v_\ell$, along $F_\ell$. Thus, the entropy of $F_\ell\otimes U_\ell$ is at least $\gamma\cdot \ell\lg \ell$.
  We may therefore set the constant $\epsilon$ in the statement of the lemma to equal the constant $\gamma$ of Conjecture~\ref{con:entropy}. This concludes the proof of Lemma~\ref{lem:subarray}.
\end{proof}

\section{The hard distribution for $L_2$-rearrangement} \label{sec:L2}

In this section we prove Lemma~\ref{lem:L2-H-lower} which says that there exists a real constant $\entropy>0$ and a fixed array $F\in \{0,1\}^n$ such that for all $\ell\in\intlens$ and all $t \in[n/4]$ such that $t \bmod 4 =0$, when $U$ is chosen uniformly at random from $\{0101,1010\}^{\frac{n}{4}}$ then the conditional entropy of the outputs $\Yv$ is
$H(\Yv\mid\Xvknown=\Xvfix)~\geq~ \entropy\cdot \ell\cdot \lg n$ for any fixed $\Xvfix$.
We begin by discussing the fixed array $F$. 

\subsection{The array $F$} \label{sec:FL2}

As for the convolution problem, for each $\ell\in\intlens$ there is a subarray of $F$ of length $\ell\lg\ell+\ell$. Each such subarray, which we denote $F_\ell$, is at distance $4\ell/\lg n+1$ from the right-hand end of $F$. This is the same high-level structure as for the convolution problem and so again, there is no overlap between the subarrays $F_\ell$ and~$F_{\ell'}$ and further, Figure~\ref{fig:F} in Section~\ref{sec:intervals} is accurate here too.

 Given any of the subarrays $F_\ell$ and an array $U_\ell$ of length $(\ell\lg \ell)$, we write $F_\ell\odot U_\ell$ to denote the $(\ell/4)$-length array that consists of the $L_2$-rearrangement distances between $U_\ell$ and every \emph{fourth} $(\ell \lg \ell)$-length substring of $F_\ell$. More precisely, for $4i\in[\ell]$, the value of $F_\ell\odot U_\ell[i]$ is the $L_2$-rearrangement distance between $F_\ell[4i,4i+\ell \lg \ell -1]$ and $U_\ell$.
 
Our main focus in this section is on proving Lemma~\ref{lem:L2subarray} which can be seen as an analogue of Lemma~\ref{lem:L2-H-lower} for a fixed length of $\ell$:

\begin{lem}
  \label{lem:L2subarray}
  There exists a real constant $\epsilon>0$ such that for all $n$ and $\ell\in\intlens$ there is a subarray $F_\ell$ for which the entropy of $F_\ell\odot U_\ell$ is at least $\epsilon\cdot\ell\lg\ell$ when $U_\ell$ is drawn uniformly at random from $\{0101,1010\}^{\frac{\ell}{4}\lg\ell}$. $F_\ell$ contains an equal number of $0$s and $1$s.
\end{lem}

In order to finish the description of the array $F$ we choose each subarray $F_\ell$ in accordance with Lemma~\ref{lem:L2subarray}. Any region of $F$ that is not part of any of the subarrays $F_\ell$ is filled with repeats of `$01$'. This ensures that these regions contain an equal number of zeros and ones (it is easily verified that each region has an even length).  This concludes the description of the array $F$.

The proof of Lemma~\ref{lem:L2-H-lower} then follows follows from Lemma~\ref{lem:L2subarray}. It is conceptually similar to the proof of Lemma~\ref{lem:conv-H-lower} for convolution which follows from Lemma~\ref{lem:subarray}. However, for the convolution problem, our proof relied on the fact that we could essentially consider the convolution of each pair $F_\ell$ and $U_\ell$ separately and simply add them up to find the convolution of $F$ and $U$. This is less immediate for  $L_2$-rearrangement distance because we need to rule out the possibility of characters from some $U_\ell$ being moved to positions in $F_{\ell'}$ for $\ell \neq \ell$. The proof (and the lower bound in general) relies on a key property of $L_2$-arrangement (proven in Lemma 3.1 from~\cite{AAB-2006}) which states that under the optimal rearrangement permutation, the $i$-th one (resp. zero) in one string is moved to the $i$-th one (resp.\@ zero) in the other. By controlling how the zeroes and ones are distributed in $U$ and $F$, we can limit how far any character is moved. 

%The details are given in the full version.

\full{
We are now ready to prove Lemma~\ref{lem:L2-H-lower}, that is lower bound the conditional entropy of $\Yv$.

\begin{proof}[Proof of Lemma~\ref{lem:L2-H-lower}]
  Let $F$ be the array described above and let $U$ be drawn uniformly at random from $\{0101,1010\}^\frac{n}{4}$. Let $\ell\in\intlens$ and $t\in[n/2]$.
  Thus, conditioned on any fixed $\Xvknown$, the distribution of $\Xv$ is uniform on $\{0101,1010\}^{\frac{\ell}{4}\lg\ell}$.
  
  Recall that $\Xv$ arrives in the stream between arrival $\inta$ and $\intb$, after which $4\ell/\lg n$ values arrive in the gap. Thus, at the beginning of the second interval, at arrival $\intc$, $\Xv$ is aligned with the $(\ell\lg\ell)$-length suffix of the subarray $F_\ell$ of $F$. Over the $\ell$ arrivals in the second interval, $\Xv$ slides along $F_\ell$ similarly to Figure~\ref{fig:slide}. We now prove that since all values in $\Xvknown$ are fixed, the outputs $\Yv$ uniquely specify $F_\ell\odot\Xv$. The analogous property for convolution was immediate. First observe that by construction the prefix of $F$ up to the start of $F_\ell$ contains an equal number of $0$s and $1$s. Similarly for $F_\ell$ itself and the suffix from $F_\ell$ to the end of $F$. Once in every four arrivals, the substring of $U$ aligned with $F$ is guaranteed (by construction) to also have an equal number of $0$s and $1$s. Therefore the $L_2$ rearrangement distance is finite. It was proven in Lemma 3.1 from~\cite{AAB-2006} that (rephrased in our notation) under the optimal rearrangement permutation, the $k$-th one (resp. zero) in $F$ is moved to the $k$-th one (resp. zero) in $U$. Therefore, every element of $U_\ell$ is moved to an element in $F_\ell$. We can therefore recover any output in $F_\ell\odot\Xv$ by taking the corresponding output in $\Yv$ and subtracting, the costs of moving the elements that are in $U$ but not in $U_\ell$. It is easily verified that as $t$ is divisible by four, the corresponding output in $\Yv$ is one of those guaranteed to have an equal number of $0$s and $1$s.
  Thus, by Lemma~\ref{lem:L2subarray}, the conditional entropy
  \[
    H(\Yv\mid\Xvknown=\Xvfix)
    ~\geq~
    \epsilon\cdot \ell\cdot \lg \ell,
    ~\geq~
    \frac{\epsilon}{4}\cdot \ell\cdot \lg n,
  \]
  since $\ell\geq n^{1/4}$. By setting the constant $\entropy$ to $\epsilon/4$ we have proved Lemma~\ref{lem:L2-H-lower}.
\end{proof}}

\subsection{The proof of Lemma~\ref{lem:L2subarray}}

In this section we prove Lemma~\ref{lem:L2subarray}. We begin by explaining the high-level approach which will make one final composition of both $F_\ell$ and $U_\ell$ into sub-arrays. For any $j\geq 0$, let  $\Usub{\ell}{j}=U_\ell[\ell \cdot j, \ell \cdot (j+1) -1]$ i.e. $\Usub{\ell}{j}$ is the $j$-th consecutive $\ell$-length sub-array of $U_\ell$. The key property that we will prove in this section is given in Lemma~\ref{lem:half} which intuitively states that given half of the bits in $U_\ell$, we can compute the other half with certainty.

\begin{lem}\label{lem:half}
Let $U_\ell$ be chosen arbitrarily from $\{0101,1010\}^\frac{\ell}{4}$ . Given $F_\ell$, $F_\ell\odot U_\ell$ and $\Usub{\ell}{2j+1}$ for all $j\geq 0$, it is possible to uniquely determine  $\Usub{\ell}{2j}$ for  all $j \geq 0$.

\end{lem}

Before we prove Lemma~\ref{lem:half}, we briefly justify why Lemma~\ref{lem:L2subarray} is in-fact a straight-forward corollary of Lemma~\ref{lem:half}. If we pick $U_\ell$ uniformly at random from $\{0101,1010\}^\frac{\ell}{4}$ then by Lemma~\ref{lem:half}, the conditional entropy, $H(F_\ell \odot U_\ell\,|\,   \Usub{\ell}{2j+1} \text{ for all } j)$ is $\Omega(\ell \lg \ell)$. This is because we always recover $\Theta(\lg \ell)$ distinct $\Usub{\ell}{2j}$, each of which is independent and has entropy $\Omega(\ell)$ bits. It then immediately follows that $H(F_\ell\odot U_\ell) \geq H(F_\ell\odot U_\ell 
\,| \, \Usub{\ell}{2j+1} \text{ for all } j)$ as required. We also require for Lemma~\ref{lem:L2subarray} that $F_\ell$ contains an equal number ones and zeros. This follows immediately from the description of $F_\ell$ below.

\begin{figure*}[htb]
    \centering
    \insertdiagram{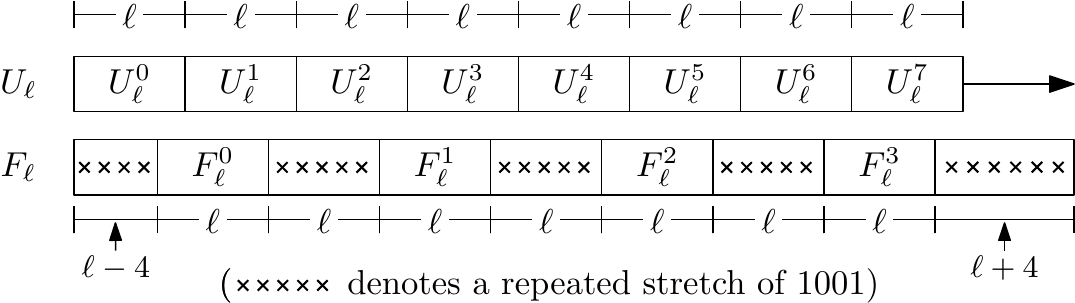}
    \caption{\label{fig:UIFI} The high-level structure of $U_\ell$ and $F_\ell$. }
\end{figure*}

\paragraph{The subarray $F_\ell$} We now give the description of $F_\ell$ which requires one final decomposition into subarrays. This description is supported by Figure~\ref{fig:UIFI}. For each $j \in \left[ \lfloor (\lg \ell)/2 \rfloor \right]$, $F_\ell$ contains a subarray $\Fsub{\ell}{j}$ of length $\ell$. Intuitively, each sub-array $\Fsub{\ell}{j}$ will be responsible for recovering $\Usub{\ell}{2j}$. These subarrays occur in order in $F_\ell$. Before and after each $\Fsub{\ell}{j}$ there are stretches of repeats of the string $1001$. Specifically, before $\Fsub{\ell}{1}$ there are $\ell/4-1$ repeats the string $1001$. Between each $\Fsub{\ell}{j}$ and $\Fsub{\ell}{j+1}$ there are $\ell/4$ repeats of the string $1001$ and after  $\Fsub{\ell}{\lfloor (\lg \ell)/2\rfloor -1}$ there are  $\ell/4+1$ repeats. These repeats of $1001$ are simply for structural padding and as we will see the contribution of these repeated $1001$ strings to the $L_2$ rearrangement distance is independent of $U_\ell$. This follows because the cost of rearranging $1001$ to match either $1010$ or $0101$ is always $2$.

The structure of $\Fsub{\ell}{j}$ is as follows $\Fsub{\ell}{j}= 10^{(2^j+3)}1^{(\ell/4-1)}0^{(\ell/4-(2^j+3))}$. Here $0^z$ (resp.\@ $1^z$) is a string of $z$ zeros (resp.\@ ones).  Intuitively, the reason that the stretch of $0$s at the start of $\Fsub{\ell}{j}$ is the exponentially increasing with $j$ is so that distance the second one in $\Fsub{\ell}{j}$ (immediately after the stretch of $0$s) is forced to move also exponentially increasing with $j$. This is will allow us to recover each $\Usub{\ell}{2j}$ from a different bit in the outputs.

\paragraph{Proof of Lemma~\ref{lem:half}} We are now in a position to prove Lemma~\ref{lem:half}. Our main focus will be on first proving that given $F_\ell$, $\Usub{\ell}{2j+1}$ for all $j$ and $F_\ell\odot U_\ell$, we can uniquely determine $\Usub{\ell}{2j}[\ell-4,\ell-1]$ for each $j\geq 0$. That is, for each $j$ whether the last four symbols of $\Usub{\ell}{2j}$ are $0101$ or $1010$. This is shown diagrammatically in Figure~\ref{fig:Dstar}. We argue that by a straight-forward repeated application of this argument we can in-fact recover the whole of $\Usub{\ell}{2j}$ for all $j\geq 0$.

\begin{figure*}[tbh]
    \centering
    \insertdiagram{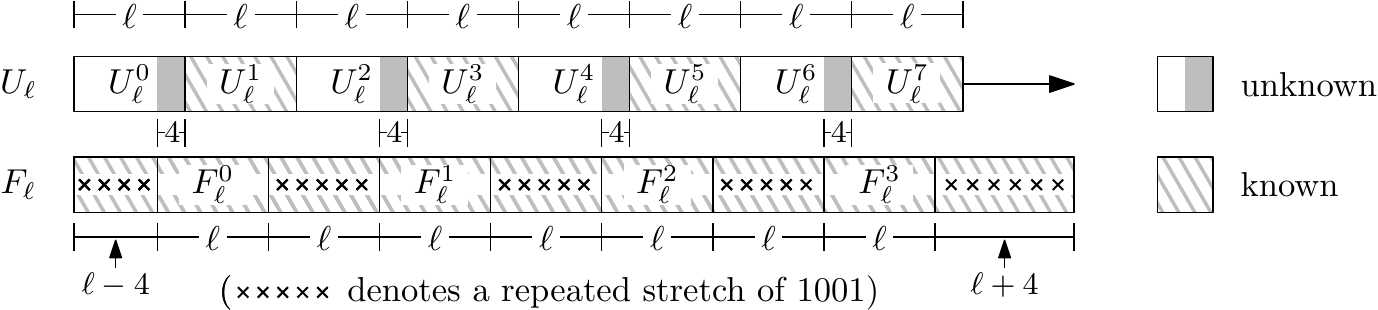}
    \caption{\label{fig:Dstar} We can determine $\Usub{\ell}{2j}[\ell-4,\ell-1]$ if we know $F_\ell$, every $\Usub{\ell}{2j+1}$ and $F_\ell\odot U_\ell$.}

\end{figure*}

We will begin by making some simplifying observations about $(F_\ell\odot U_\ell)[0]$. Recall that $(F_\ell\odot U_\ell)[0]$ was defined to be the $L_2$ rearrangement distance between $F_\ell[0,|U_\ell| -1]$ and $U_\ell$. The first observation is that is finite because both strings contain an equal number of zeros and ones. 

The $L_2$ rearrangement distance $(F_\ell\odot U_\ell)[0]$ can be expressed as the sum of the \emph{contributions} from moving each $U_\ell[i]$, over all $i \in [m]$.
Let the contribution of $U_\ell[i]$, denoted, $\contrib(i)$ be the square of the distance that $U_\ell[i]$ is moved under the optimal rearrangement. We then have that $(F_\ell\odot U_\ell)[0]= \sum_i \contrib(j)$. Finally, we let $D^\star$ be the sum of the contributions of the locations in every $\Usub{\ell}{2j}[\ell-4,\ell-1]$, i.e.\@
 $D^\star= \sum_{j} \sum_{k=0}^{3} \left(\contrib(2j \cdot \ell +(\ell-4) + k \right)$. We will also refer to the contribution of a substring which is defined naturally to be the sum of the contributions of its constituent characters. For example the contribution of $\Usub{\ell}{j}$ equals $\sum \left\{ \contrib(r) \st r \in [\ell \cdot j, \ell \cdot (j+1) -1] \right\}.$

Our proof will be in two stages. First we will prove in Lemma~\ref{lem:compD} that we can compute $D^\star$ from $F_\ell$, $F_\ell\odot U_\ell$ and $\Usub{\ell}{2j+1}$ for all $j\geq 0$. Second we will prove that for any $j>0$, we can determine $\Usub{\ell}{2j}[\ell-4,\ell-1]$ from $D^\star$.

In the proof of Lemma~\ref{lem:compD} we argue that $D^\star$ can be calculated directly from $(F_\ell\odot U_\ell)[0]$ by subtracting the contributions of $\Usub{\ell}{2j+1}$  and $\Usub{\ell}{2j}[0,\ell-5]$ for all $j\geq 0$. More specifically, we will prove that the contribution of any $\Usub{\ell}{2j+1}$ can calculated from $\Usub{\ell}{2j+1}$ and $F_\ell$, which are both known. In particular, the contribution of any $\Usub{\ell}{2j+1}$ is independent of every unknown $\Usub{\ell}{2j}$. Further, we will prove that although $\Usub{\ell}{2j}$ is unknown, the contribution of $\Usub{\ell}{2j}[0,\ell-5]$, always equals $\ell/2-2$, regardless of the choice of $U_\ell$. 
\begin{lem}\label{lem:compD}
$D^\star$ can be computed from $F_\ell$, $F_\ell\odot U_\ell$ and $\Usub{\ell}{2j+1}$ for all $j\geq 0$.
\end{lem}

\full{\begin{proof}
The value of $D^\star$ is calculated directly from $(F_\ell\odot U_\ell)[0]$ by subtracting the contributions of $\Usub{\ell}{2j+1}$  and $\Usub{\ell}{2j}[0,\ell-5]$ for all $j\geq 0$. We will now prove that for any $j$, we can calculate the contribution of $\Usub{\ell}{2j+1}$ and that contribution of $\Usub{\ell}{2j}[0,\ell-5]$, in-fact always equals $\ell/2-2$, regardless of the choice of $U_\ell$. 

In this proof we rely heavily on Lemma 3.1 from~\cite{AAB-2006} which states that under the optimal rearrangement permutation, the $i$-th one (resp. zero) in $U_\ell$ is moved to the $i$-th one (resp. zero) in $F_\ell[0,|U_\ell| -1]$. 

For any $j$, consider,  $\Usub{\ell}{2j}$ and $\Usub{\ell}{2j+1}$. The number of ones in $\Usub{\ell}{2j}$ (resp. $\Usub{\ell}{2j+1}$) is fixed, independent of the choice of $U_\ell$. In particular there are exactly $\ell/2$ zeroes and $\ell/2$ ones.  It is easily verified that, by construction, $F_\ell[2j\cdot \ell ,(2j+2)  \cdot \ell -1]$ also contains exactly $\ell$ zeros and $\ell$ ones. Therefore, the $i$-th one (resp.\@ zero) in $\Usub{\ell}{2j}$ is moved to the $i$-th one (resp.\@ zero) in  $F_\ell[2j\cdot \ell ,(2j+2)  \cdot \ell -1]$. Similarly, the $i$-th one (resp.\@ zero) in $\Usub{\ell}{2j}$ is moved to the $(i+\ell/2)$-th one (resp.\@ zero) in  $F_\ell[2j\cdot \ell ,(2j+2)  \cdot \ell -1]$

Consider any $\Usub{\ell}{2j+1}$ which is known. By the above observation, we can therefore determine which position in $F_\ell[2j\cdot \ell ,(2j+2)  \cdot \ell -1]$, each character in $\Usub{\ell}{2j+1}$ is moved to under the optimal rearrangement. From this we can then directly compute the contribution of each $\Usub{\ell}{2j+1}$ to $(F_\ell\odot U_\ell)[0]$.

Consider any $\Usub{\ell}{2j}$ which is unknown. As observed above, the $i$-th one (resp.\@ zero) in $\Usub{\ell}{2j}$ is moved to the $i$-th one (resp.\@ zero) in  $F_\ell[2j\cdot \ell ,(2j+2)  \cdot \ell -1]$. By construction, we have that  $F_\ell[2j\cdot \ell,(2j+1) \cdot \ell -5]$ consists entirely  of repeats of $1001$.  Further for any $i$, we have that $\Usub{\ell}{2j}[4i,4i+3]$ is either $1010$ or $0101$. Therefore for all $i<\ell/4$ we have that the two ones (resp.\@ zeroes) in $\Usub{\ell}{2j}[4i,4i+3]$ are moved to the two ones (resp.\@ zeroes) in $F_\ell[2j \cdot \ell + 4i,2j \cdot \ell + 4i+3]=1001$. The key observation is that regardless of whether $\Usub{\ell}{2j}[4i,4i+3]=1010$ or $0101$, the contribution of $\Usub{\ell}{2j}[4i,4i+3]$ is $2$. Therefore for any $U_\ell$, the contribution of  $\Usub{\ell}{2j}[0,\ell-5]$ is $\ell/2-2$.
\end{proof}}

In Lemma 11 we will prove that we can compute $\Usub{\ell}{2j}[\ell-4,\ell-1]$ from $D^\star$ (for any sufficiently large $j$). The intuition behind this is given by Fact~\ref{lem:contU} which gives an explicit formula for the contribution of $\Usub{\ell}{2j}[\ell-4,\ell-1]$. Observe that the contribution depends only on whether  $\Usub{\ell}{2j}[\ell-4,\ell-1]$ equals $1010$ ($v_j=1$) or $0101$ ($v_j=0$). The intuition is that we can extract $v_j$ from the $(j+1)$-th bit of $D^\star$.

\begin{fact}\label{lem:contU}
For any $j$, let  $v_j=1$ if $\Usub{\ell}{2j}[\ell-4,\ell-1]=1010$ and $v_j=0$ otherwise. The contribution of $\Usub{\ell}{2j}[\ell-4,\ell-1]$ is exactly $v_j \cdot 2^{j+1} + 2^{2j}+ 2$.
\end{fact}

\full{\begin{proof}
We begin by arguing that under the optimal rearrangement permutation, the two ones (resp.\@ zeroes) in $\Usub{\ell}{2j}[\ell-4,\ell-1]$ are moved to the leftmost two ones (resp.\@ zeroes) in $\Fsub{\ell}{j}$. We will again rely heavily on Lemma 3.1 from~\cite{AAB-2006} which states that under the optimal rearrangement permutation, the $i$-th one (resp. zero) in $U_\ell$ is moved to the $i$-th one (resp. zero) in $F_\ell[0,|U_\ell| -1]$.

In the proof of Lemma~\ref{lem:compD}, we argued that the $i$-th one (resp.\@ zero) in $\Usub{\ell}{2j}$ is moved to the $i$-th one (resp.\@ zero) in  $F_\ell[2j\cdot \ell ,(2j+2)  \cdot \ell -1]$. It is easily verified that $F_\ell[2j\cdot \ell ,(2j+2)  \cdot \ell -1]$ consists of exactly $\ell/4-1$ repeats of $1001$ followed by $\Fsub{\ell}{j}$. Therefore as $\Usub{\ell}{2j}[0,\ell-5]$
contains exactly $\ell/2-2$ ones (resp.\@ zeroes), the two ones in  $\Usub{\ell}{2j}[\ell-4,\ell-1]$ are indeed moved to the leftmost two ones (resp.\@ zeroes) in $\Fsub{\ell}{j}$.

\begin{figure*}[ht]
    \centering
    \insertdiagram{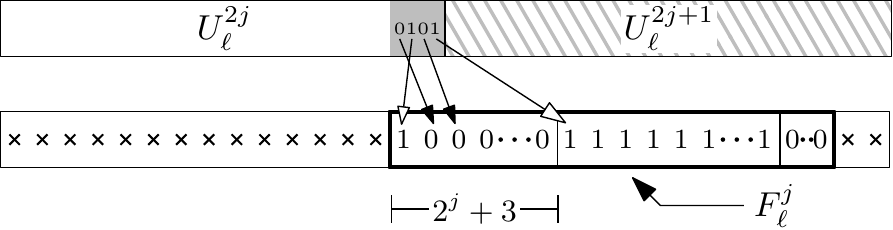}
    \caption{\label{fig:zd} The rearrangement of the symbols in $\Usub{\ell}{2j}[\ell-4,\ell-1]$ under the optimal rearrangement permutation. The highlighted region is $\Fsub{\ell}{j}$.}
\end{figure*}

We now argue about the contribution of each character in $\Usub{\ell}{2j}[\ell-4,\ell-1]$ in turn. This argument is  is supported by Figure~\ref{fig:zd}.

First consider the first one in $\Usub{\ell}{2j}[\ell-4,\ell-1]$ which is moved to the first one in $\Fsub{\ell}{j}$. By construction we have that, when indexed from the start of $F_\ell$, the leftmost one in $\Fsub{\ell}{j}$ is at position $x=(2j+1)\cdot \ell -4.$ Similarly, when indexed from the start of $U_{\ell}$, the leftmost one in $\Usub{\ell}{2j}[\ell-4,\ell-1]$ is at position $x + (1-v_j)$. Therefore the contribution of the first one in $\Usub{\ell}{2j}$ is exactly $(1-v_j)^2.$

The first zero  in $\Usub{\ell}{2j}[\ell-4,\ell-1]$ is moved to the first zero in $\Fsub{\ell}{j}$. The first zero in $\Usub{\ell}{2j}[\ell-4,\ell-1]$ is at position $x + v_j.$ The first zero in $\Fsub{\ell}{j}$ is at position $x+1.$ Therefore the contribution of the first zero  in $\Usub{\ell}{2j}$  is also exactly $(1-v_j)^2$.

The second zero  in $\Usub{\ell}{2j}[\ell-4,\ell-1]$ is moved to the second zero in $\Fsub{\ell}{j}$. The second zero in $\Usub{\ell}{2j}[\ell-4,\ell-1]$ is at position $x + 2+v_j$ and the second zero in $\Fsub{\ell}{j}$ is at position $x +2$. Therefore the contribution of the first zero  in $\Usub{\ell}{2j}$  is exactly $v_j^2$. 

Finally, the second one in $\Usub{\ell}{2j}[\ell-4,\ell-1]$ is moved to the second  one  in $\Fsub{\ell}{j}$. Again, by construction, we have that the second one in $\Fsub{\ell}{j}$ is at position $x + 3 + 2^j$ (indexed from the start of  $F_\ell$). Similarly, the second one in $\Usub{\ell}{2j}[\ell-4,\ell-1]$ is at position $x + 3-v_j$. Therefore the contribution of the second one  in $\Usub{\ell}{2j}$  is exactly, $(2^{j}+v_j)^2$. Summing over all four characters we have that the total contribution of  $\Usub{\ell}{2j}[\ell-4,\ell-1]$ is:

\[ (2^{j}+v_j)^2 + v_j^2 + 2(1-v_j)^2 \]

\noindent Expanding and simplifying we have that this is

\[ (2^{2j}+ 2) + v_j \cdot 2^{j+1} +4v_j^2 -4v_j \]

\noindent As $v_j \in \{0,1\}$ we have that $v_j^2 = v_j$ so this simplifies further to:

\[  v_j \cdot 2^{j+1} + 2^{2j}+ 2 \]
\end{proof}}

\noindent We can now prove Lemma~\ref{lem:fromD} which follows almost immediately from Fact~\ref{lem:contU}. 

\begin{lem}\label{lem:fromD}
For any $j \geq 0$, it is possible to compute $\Usub{\ell}{2j}[\ell-4,\ell-1]$ from $D^\star$.
\end{lem}

\full{\begin{proof}

Let $D_2^\star$ equal $D^\star - \sum_j (2^{2j} + 2)$ which can be calculated directly from $D^\star$. An alternative and equivalent definition of $D_2^\star$ follows from Fact~\ref{lem:contU} and is given by \[ D_2^\star = \sum_j v_j \cdot 2^{j+1}. \]

We can therefore compute $v_j$ and hence $\Usub{\ell}{2j}[\ell-4,\ell-1]$ by inspecting the $(j+1)$-th bit of $D_2^\star$.

\end{proof}}

\full{\paragraph{Recovering the rest of $U_{\ell,(2j)}$} So far we have only proven that we can recover $\Usub{\ell}{2j}[\ell-4,\ell-1]$  for all $j$. The claim that we can in-fact recover the whole of $\Usub{\ell}{2j}$ follows by repeatedly application of the the proof above. Specifically, once we have recovered $\Usub{\ell}{2j}[\ell-4,\ell-1]$ for all $j$, we can use this additional information (and $(F_\ell\odot U_\ell)[1]$ instead of  $(F_\ell\odot U_\ell)[0]$) to recover $\Usub{\ell}{2j}[\ell-8,\ell-5]$ for all $j$ and so on. More formally we proceed by induction on increasing $k$ by observing that using the above argument given  $F_\ell$, $(F_\ell\odot U_\ell)[k]$, $\Usub{\ell}{2j+1}$ for all $j\geq 0$ and  $\Usub{\ell}{2j+1}[\ell-4k,\ell-1]$ for all $j\geq 0$ we can recover $\Usub{\ell}{2j+1}[\ell-4k-4,\ell-4k-1]$ for all $j$.}

\bibliographystyle{plain-fi}
\bibliography{longnames,bib-latest,bristol}

\end{document}